\documentclass{fac}

\usepackage{graphicx}
\usepackage{amssymb}
\usepackage{amsfonts}
\usepackage{amsmath}
\usepackage{dsfont}
\usepackage{MnSymbol}
\usepackage{mathtools}
\usepackage{epsfig}
\usepackage{epstopdf}

\newtheorem{theorem}{Theorem}[section]
\newtheorem{definition}[theorem]{Definition}
\newtheorem{proposition}[theorem]{Proposition}

\title[Draft of A Calculus of Truly Concurrent Mobile Processes]
      {A Calculus of Truly Concurrent Mobile Processes}

\author[Yong Wang]
    {Yong Wang\\
     College of Computer Science and Technology,\\
     Faculty of Information Technology,\\
     Beijing University of Technology, Beijing, China\\
     }

\correspond{Yong Wang, Pingleyuan 100, Chaoyang District, Beijing, China.
            e-mail: wangy@bjut.edu.cn}

\pubyear{2017}
\pagerange{\pageref{firstpage}--\pageref{lastpage}}

\begin{document}
\label{firstpage}

\makecorrespond

\maketitle

\begin{abstract}
We make a mixture of Milner's $\pi$-calculus and our previous work on truly concurrent process algebra, which is called $\pi_{tc}$. We introduce syntax and semantics of $\pi_{tc}$, its properties based on strongly truly concurrent bisimilarities. Also, we include an axiomatization of $\pi_{tc}$. $\pi_{tc}$ can be used as a formal tool in verifying mobile systems in a truly concurrent flavor.
\end{abstract}


\section{Introduction}\label{int}

The famous work on parallelism and concurrency \cite{CM} is bisimilarity and its process algebra. CCS (A Calculus of Communicating Systems) \cite{CCS} \cite{CC} is a calculus based on bisimulation semantics model. CCS has good semantic properties based on the interleaving bisimulation. ACP (Algebra of Communicating Systems) \cite{ACP} is an axiomatization for several computational properties based on bisimulation, includes sequential and alternative computation, parallelism and communication, encapsulation, recursion, and abstraction. $\pi$-calculus is an extension of CCS, and aims at mobile processes.

The other camp of concurrency is true concurrency. The researches on true concurrency are still active. Firstly, there are several truly concurrent bisimulations, the representatives are: pomset bisimulation, step bisimulation, history-preserving (hp-) bisimulation, and especially hereditary history-preserving (hhp-) bisimulation \cite{HHP1} \cite{HHP2}. These truly concurrent bisimulations are studied in different structures \cite{ES1} \cite{ES2} \cite{CM}: Petri nets, event structures, domains, and also a uniform form called TSI (Transition System with Independence) \cite{SFL}. There are also several logics based on different truly concurrent bisimulation equivalences, for example, SFL (Separation Fixpoint Logic) and TFL (Trace Fixpoint Logic) \cite{SFL} are extensions on true concurrency of mu-calculi \cite{MUC} on bisimulation equivalence, and also a logic with reverse modalities \cite {RL1} \cite{RL2} based on the so-called reverse bisimulations with a reverse flavor. Recently, a uniform logic for true concurrency \cite{LTC1} \cite{LTC2} was represented, which used a logical framework to unify several truly concurrent bisimulations, including pomset bisimulation, step bisimulation, hp-bisimulation and hhp-bisimulation.

There are simple comparisons between HM logic and bisimulation, as the uniform logic \cite{LTC1} \cite{LTC2} and truly concurrent bisimulations; the algebraic laws \cite{ALNC}, ACP \cite{ACP} and bisimulation, as the algebraic laws APTC \cite{ATC} and truly concurrent bisimulations; CCS and bisimulation, as the calculus CTC \cite{CTC} and truly concurrent bisimulations; $\pi$-calculus and bisimulation, as true concurrency and \emph{what}, which is still missing.

In this paper, we design a calculus of truly concurrent mobile processes ($\pi_{tc}$) following the way paved by $\pi$-calculus for bisimulation and our previous work on truly concurrent process algebra CTC \cite{CTC} and APTC \cite{ATC}. This paper is organized as follows. In section \ref{bac}, we introduce some preliminaries, including a brief introduction to $\pi$-calculus, and also preliminaries on true concurrency. We introduce the syntax and operational semantics of $\pi_{tc}$ in section \ref{sos}, its properties for strongly truly concurrent bisimulations in section \ref{stcb}, its axiomatization in section \ref{at}. Finally, in section \ref{con}, we conclude this paper.

\section{Backgrounds}\label{bac}

\subsection{$\pi$-calculus}

$\pi$-calculus \cite{PI1} \cite{PI2} is a calculus for mobile processes, which have changing structure. The component processes not only can be arbitrarily linked, but also can change the linkages by communications among them. $\pi$-calculus is an extension of the process algebra CCS \cite{CCS}:

\begin{itemize}
  \item It treats names, variables and substitutions more carefully, since names may be free or bound;
  \item Names are mobile by references, rather that by values;
  \item There are three kinds of prefixes, $\tau$ prefix $\tau.P$, output prefix $\overline{x}y.P$ and input prefix $x(y).P$, which are most distinctive to CCS;
  \item Since strong bisimilarity is not preserved by substitutions because of its interleaving nature, $\pi$-calculus develops several kinds of strong bisimulations, and discusses their laws modulo these bisimulations.
\end{itemize}

\subsection{True Concurrency}

The related concepts on true concurrency are defined based on the following concepts.

\begin{definition}[Prime event structure with silent event]\label{PES}
Let $\Lambda$ be a fixed set of labels, ranged over $a,b,c,\cdots$ and $\tau$. A ($\Lambda$-labelled) prime event structure with silent event $\tau$ is a tuple $\mathcal{E}=\langle \mathbb{E}, \leq, \sharp, \lambda\rangle$, where $\mathbb{E}$ is a denumerable set of events, including the silent event $\tau$. Let $\hat{\mathbb{E}}=\mathbb{E}\backslash\{\tau\}$, exactly excluding $\tau$, it is obvious that $\hat{\tau^*}=\epsilon$, where $\epsilon$ is the empty event. Let $\lambda:\mathbb{E}\rightarrow\Lambda$ be a labelling function and let $\lambda(\tau)=\tau$. And $\leq$, $\sharp$ are binary relations on $\mathbb{E}$, called causality and conflict respectively, such that:

\begin{enumerate}
  \item $\leq$ is a partial order and $\lceil e \rceil = \{e'\in \mathbb{E}|e'\leq e\}$ is finite for all $e\in \mathbb{E}$. It is easy to see that $e\leq\tau^*\leq e'=e\leq\tau\leq\cdots\leq\tau\leq e'$, then $e\leq e'$.
  \item $\sharp$ is irreflexive, symmetric and hereditary with respect to $\leq$, that is, for all $e,e',e''\in \mathbb{E}$, if $e\sharp e'\leq e''$, then $e\sharp e''$.
\end{enumerate}

Then, the concepts of consistency and concurrency can be drawn from the above definition:

\begin{enumerate}
  \item $e,e'\in \mathbb{E}$ are consistent, denoted as $e\frown e'$, if $\neg(e\sharp e')$. A subset $X\subseteq \mathbb{E}$ is called consistent, if $e\frown e'$ for all $e,e'\in X$.
  \item $e,e'\in \mathbb{E}$ are concurrent, denoted as $e\parallel e'$, if $\neg(e\leq e')$, $\neg(e'\leq e)$, and $\neg(e\sharp e')$.
\end{enumerate}
\end{definition}

The prime event structure without considering silent event $\tau$ is the original one in \cite{ES1} \cite{ES2} \cite{CM}.

\begin{definition}[Configuration]
Let $\mathcal{E}$ be a PES. A (finite) configuration in $\mathcal{E}$ is a (finite) consistent subset of events $C\subseteq \mathcal{E}$, closed with respect to causality (i.e. $\lceil C\rceil=C$). The set of finite configurations of $\mathcal{E}$ is denoted by $\mathcal{C}(\mathcal{E})$. We let $\hat{C}=C\backslash\{\tau\}$.
\end{definition}

Usually, truly concurrent behavioral equivalences are defined by events $e\in\mathcal{E}$ and prime event structure $\mathcal{E}$ (see related concepts in section \ref{STCC}), in contrast to interleaving behavioral equivalences by actions $a,b\in\mathcal{P}$ and process (graph) $\mathcal{P}$. Indeed, they have correspondences, in \cite{SFL}, models of concurrency, including Petri nets, transition systems and event structures, are unified in a uniform representation -- TSI (Transition System with Independence).

If $x$ is a process, let $C(x)$ denote the corresponding configuration (the already executed part of the process $x$, of course, it is free of conflicts), when $x\xrightarrow{e} x'$, the corresponding configuration $C(x)\xrightarrow{e}C(x')$ with $C(x')=C(x)\cup\{e\}$, where $e$ may be caused by some events in $C(x)$ and concurrent with the other events in $C(x)$, or entirely concurrent with all events in $C(x)$, or entirely caused by all events in $C(x)$. Though the concurrent behavioral equivalences (Definition \ref{PSB} and \ref{HHPB}) are defined based on configurations (pasts of processes), they can also be defined based on processes (futures of configurations), we omit the concrete definitions.

With a little abuse of concepts, in the following of the paper, we will not distinguish actions and events, prime event structures and processes, also concurrent behavior equivalences based on configurations and processes, and use them freely, unless they have specific meanings. 

\section{Syntax and Operational Semantics}\label{sos}

We assume an infinite set $\mathcal{N}$ of (action or event) names, and use $a,b,c,\cdots$ to range over $\mathcal{N}$, use $x,y,z,w,u,v$ as meta-variables over names. We denote by $\overline{\mathcal{N}}$ the set of co-names and let $\overline{a},\overline{b},\overline{c},\cdots$ range over $\overline{\mathcal{N}}$. Then we set $\mathcal{L}=\mathcal{N}\cup\overline{\mathcal{N}}$ as the set of labels, and use $l,\overline{l}$ to range over $\mathcal{L}$. We extend complementation to $\mathcal{L}$ such that $\overline{\overline{a}}=a$. Let $\tau$ denote the silent step (internal action or event) and define $Act=\mathcal{L}\cup\{\tau\}$ to be the set of actions, $\alpha,\beta$ range over $Act$. And $K,L$ are used to stand for subsets of $\mathcal{L}$ and $\overline{L}$ is used for the set of complements of labels in $L$.

Further, we introduce a set $\mathcal{X}$ of process variables, and a set $\mathcal{K}$ of process constants, and let $X,Y,\cdots$ range over $\mathcal{X}$, and $A,B,\cdots$ range over $\mathcal{K}$. For each process constant $A$, a nonnegative arity $ar(A)$ is assigned to it. Let $\widetilde{x}=x_1,\cdots,x_{ar(A)}$ be a tuple of distinct name variables, then $A(\widetilde{x})$ is called a process constant. $\widetilde{X}$ is a tuple of distinct process variables, and also $E,F,\cdots$ range over the recursive expressions. We write $\mathcal{P}$ for the set of processes. Sometimes, we use $I,J$ to stand for an indexing set, and we write $E_i:i\in I$ for a family of expressions indexed by $I$. $Id_D$ is the identity function or relation over set $D$. The symbol $\equiv_{\alpha}$ denotes equality under standard alpha-convertibility, note that the subscript $\alpha$ has no relation to the action $\alpha$.

\subsection{Syntax}

We use the Prefix $.$ to model the causality relation $\leq$ in true concurrency, the Summation $+$ to model the conflict relation $\sharp$ in true concurrency, and the Composition $\parallel$ to explicitly model concurrent relation in true concurrency. And we follow the conventions of process algebra.

\begin{definition}[Syntax]\label{syntax}
A truly concurrent process $P$ is defined inductively by the following formation rules:

\begin{enumerate}
  \item $A(\widetilde{x})\in\mathcal{P}$;
  \item $\textbf{nil}\in\mathcal{P}$;
  \item if $P\in\mathcal{P}$, then the Prefix $\tau.P\in\mathcal{P}$, for $\tau\in Act$ is the silent action;
  \item if $P\in\mathcal{P}$, then the Output $\overline{x}y.P\in\mathcal{P}$, for $x,y\in Act$;
  \item if $P\in\mathcal{P}$, then the Input $x(y).P\in\mathcal{P}$, for $x,y\in Act$;
  \item if $P\in\mathcal{P}$, then the Restriction $(x)P\in\mathcal{P}$, for $x\in Act$;
%  \item if $P\in\mathcal{P}$, then the Match $[x=y]P\in\mathcal{P}$, for $x,y\in Act$;
  \item if $P,Q\in\mathcal{P}$, then the Summation $P+Q\in\mathcal{P}$;
  \item if $P,Q\in\mathcal{P}$, then the Composition $P\parallel Q\in\mathcal{P}$;
\end{enumerate}

The standard BNF grammar of syntax of $\pi_{tc}$ can be summarized as follows:

$$P::=A(\widetilde{x})\quad|\quad\textbf{nil}\quad|\quad\tau.P\quad|\quad \overline{x}y.P\quad |\quad x(y).P \quad|\quad (x)P \quad | \quad P+P\quad |\quad P\parallel P.$$
\end{definition}

In $\overline{x}y$, $x(y)$ and $\overline{x}(y)$, $x$ is called the subject, $y$ is called the object and it may be free or bound.

\begin{definition}[Free variables]
The free names of a process $P$, $fn(P)$, are defined as follows.

\begin{enumerate}
  \item $fn(A(\widetilde{x}))\subseteq\{\widetilde{x}\}$;
  \item $fn(\textbf{nil})=\emptyset$;
  \item $fn(\tau.P)=fn(P)$;
  \item $fn(\overline{x}y.P)=fn(P)\cup\{x\}\cup\{y\}$;
  \item $fn(x(y).P)=fn(P)\cup\{x\}-\{y\}$;
  \item $fn((x)P)=fn(P)-\{x\}$;
%  \item $fn([x=y]P)=fn(P)$;
  \item $fn(P+Q)=fn(P)\cup fn(Q)$;
  \item $fn(P\parallel Q)=fn(P)\cup fn(Q)$.
\end{enumerate}
\end{definition}

\begin{definition}[Bound variables]
Let $n(P)$ be the names of a process $P$, then the bound names $bn(P)=n(P)-fn(P)$.
\end{definition}

For each process constant schema $A(\widetilde{x})$, a defining equation of the form

$$A(\widetilde{x})\overset{\text{def}}{=}P$$

is assumed, where $P$ is a process with $fn(P)\subseteq \{\widetilde{x}\}$.

\begin{definition}[Substitutions]\label{subs}
A substitution is a function $\sigma:\mathcal{N}\rightarrow\mathcal{N}$. For $x_i\sigma=y_i$ with $1\leq i\leq n$, we write $\{y_1/x_1,\cdots,y_n/x_n\}$ or $\{\widetilde{y}/\widetilde{x}\}$ for $\sigma$. For a process $P\in\mathcal{P}$, $P\sigma$ is defined inductively as follows:
\begin{enumerate}
  \item if $P$ is a process constant $A(\widetilde{x})=A(x_1,\cdots,x_n)$, then $P\sigma=A(x_1\sigma,\cdots,x_n\sigma)$;
  \item if $P=\textbf{nil}$, then $P\sigma=\textbf{nil}$;
  \item if $P=\tau.P'$, then $P\sigma=\tau.P'\sigma$;
  \item if $P=\overline{x}y.P'$, then $P\sigma=\overline{x\sigma}y\sigma.P'\sigma$;
  \item if $P=x(y).P'$, then $P\sigma=x\sigma(y).P'\sigma$;
  \item if $P=(x)P'$, then $P\sigma=(x\sigma)P'\sigma$;
%  \item if $P=[x=y]P'$, then $P\sigma=[x\sigma=y\sigma]P'\sigma$;
  \item if $P=P_1+P_2$, then $P\sigma=P_1\sigma+P_2\sigma$;
  \item if $P=P_1\parallel P_2$, then $P\sigma=P_1\sigma \parallel P_2\sigma$.
\end{enumerate}
\end{definition}

\subsection{Operational Semantics}

The operational semantics is defined by LTSs (labelled transition systems), and it is detailed by the following definition.

\begin{definition}[Semantics]\label{semantics}
The operational semantics of $\pi_{tc}$ corresponding to the syntax in Definition \ref{syntax} is defined by a series of transition rules, named $\textbf{ACT}$, $\textbf{SUM}$, $\textbf{IDE}$, $\textbf{PAR}$, $\textbf{COM}$ and $\textbf{CLOSE}$, $\textbf{RES}$ and $\textbf{OPEN}$ indicate that the rules are associated respectively with Prefix, Summation, Match, Identity, Parallel Composition, Communication, and Restriction in Definition \ref{syntax}. They are shown in Table \ref{TRForPITC}.

\begin{center}
    \begin{table}
        \[\textbf{TAU-ACT}\quad \frac{}{\tau.P\xrightarrow{\tau}P} \quad \textbf{OUTPUT-ACT}\quad \frac{}{\overline{x}y.P\xrightarrow{\overline{x}y}P}\]

        \[\textbf{INPUT-ACT}\quad \frac{}{x(z).P\xrightarrow{x(w)}P\{w/z\}}\quad (w\notin fn((z)P))\]

        \[\textbf{PAR}_1\quad \frac{P\xrightarrow{\alpha}P'\quad Q\nrightarrow}{P\parallel Q\xrightarrow{\alpha}P'\parallel Q}\quad (bn(\alpha)\cap fn(Q)=\emptyset) \quad \textbf{PAR}_2\quad \frac{Q\xrightarrow{\alpha}Q'\quad P\nrightarrow}{P\parallel Q\xrightarrow{\alpha}P\parallel Q'}\quad (bn(\alpha)\cap fn(P)=\emptyset)\]

        \[\textbf{PAR}_3\quad \frac{P\xrightarrow{\alpha}P'\quad Q\xrightarrow{\beta}Q'}{P\parallel Q\xrightarrow{\{\alpha,\beta\}}P'\parallel Q'}\quad (\beta\neq\overline{\alpha}, bn(\alpha)\cap bn(\beta)=\emptyset, bn(\alpha)\cap fn(Q)=\emptyset,bn(\beta)\cap fn(P)=\emptyset)\]

        \[\textbf{PAR}_4\quad \frac{P\xrightarrow{x_1(z)}P'\quad Q\xrightarrow{x_2(z)}Q'}{P\parallel Q\xrightarrow{\{x_1(w),x_2(w)\}}P'\{w/z\}\parallel Q'\{w/z\}}\quad (w\notin fn((z)P)\cup fn((z)Q))\]

        \[\textbf{COM}\quad \frac{P\xrightarrow{\overline{x}y}P'\quad Q\xrightarrow{x(z)}Q'}{P\parallel Q\xrightarrow{\tau}P'\parallel Q'\{y/z\}}\]

        \[\textbf{CLOSE}\quad \frac{P\xrightarrow{\overline{x}(w)}P'\quad Q\xrightarrow{x(w)}Q'}{P\parallel Q\xrightarrow{\tau}(w)(P'\parallel Q')}\]

        \[\textbf{SUM}_1\quad \frac{P\xrightarrow{\alpha}P'}{P+Q\xrightarrow{\alpha}P'} \quad \textbf{SUM}_2\quad \frac{P\xrightarrow{\{\alpha_1,\cdots,\alpha_n\}}P'}{P+Q\xrightarrow{\{\alpha_1,\cdots,\alpha_n\}}P'}\]

%        \[\textbf{MATCH}_1\quad \frac{P\xrightarrow{\alpha}P'}{[x=x]P\xrightarrow{\alpha}P'} \quad
%        \textbf{MATCH}_2\quad \frac{P\xrightarrow{\{\alpha_1,\cdots,\alpha_n\}}P'}{[x=x]P\xrightarrow{\{\alpha_1,\cdots,\alpha_n\}}P'}\]

        \[\textbf{IDE}_1\quad\frac{P\{\widetilde{y}/\widetilde{x}\}\xrightarrow{\alpha}P'}{A(\widetilde{y})\xrightarrow{\alpha}P'}\quad (A(\widetilde{x})\overset{\text{def}}{=}P) \quad \textbf{IDE}_2\quad\frac{P\{\widetilde{y}/\widetilde{x}\}\xrightarrow{\{\alpha_1,\cdots,\alpha_n\}}P'} {A(\widetilde{y})\xrightarrow{\{\alpha_1,\cdots,\alpha_n\}}P'}\quad (A(\widetilde{x})\overset{\text{def}}{=}P)\]

        \[\textbf{RES}_1\quad \frac{P\xrightarrow{\alpha}P'}{(y)P\xrightarrow{\alpha}(y)P'}\quad (y\notin n(\alpha)) \quad \textbf{RES}_2\quad \frac{P\xrightarrow{\{\alpha_1,\cdots,\alpha_n\}}P'}{(y)P\xrightarrow{\{\alpha_1,\cdots,\alpha_n\}}(y)P'}\quad (y\notin n(\alpha_1)\cup\cdots\cup n(\alpha_n))\]

        \[\textbf{OPEN}_1\quad \frac{P\xrightarrow{\overline{x}y}P'}{(y)P\xrightarrow{\overline{x}(w)}P'\{w/y\}} \quad (y\neq x, w\notin fn((y)P'))\]

        \[\textbf{OPEN}_2\quad \frac{P\xrightarrow{\{\overline{x}_1 y,\cdots,\overline{x}_n y\}}P'}{(y)P\xrightarrow{\{\overline{x}_1(w),\cdots,\overline{x}_n(w)\}}P'\{w/y\}} \quad (y\neq x_1\neq\cdots\neq x_n, w\notin fn((y)P'))\]

        \caption{Transition rules of $\pi_{tc}$}
        \label{TRForPITC}
    \end{table}
\end{center}
\end{definition}

\subsection{Properties of Transitions}

\begin{proposition}
\begin{enumerate}
  \item If $P\xrightarrow{\alpha}P'$ then
  \begin{enumerate}
    \item $fn(\alpha)\subseteq fn(P)$;
    \item $fn(P')\subseteq fn(P)\cup bn(\alpha)$;
  \end{enumerate}
  \item If $P\xrightarrow{\{\alpha_1,\cdots,\alpha_n\}}P'$ then
  \begin{enumerate}
    \item $fn(\alpha_1)\cup\cdots\cup fn(\alpha_n)\subseteq fn(P)$;
    \item $fn(P')\subseteq fn(P)\cup bn(\alpha_1)\cup\cdots\cup bn(\alpha_n)$.
  \end{enumerate}
\end{enumerate}
\end{proposition}

\begin{proof}
By induction on the depth of inference.
\end{proof}

\begin{proposition}
Suppose that $P\xrightarrow{\alpha(y)}P'$, where $\alpha=x$ or $\alpha=\overline{x}$, and $x\notin n(P)$, then there exists some $P''\equiv_{\alpha}P'\{z/y\}$, $P\xrightarrow{\alpha(z)}P''$.
\end{proposition}

\begin{proof}
By induction on the depth of inference.
\end{proof}

\begin{proposition}
If $P\rightarrow P'$, $bn(\alpha)\cap fn(P'\sigma)=\emptyset$, and $\sigma\lceil bn(\alpha)=id$, then there exists some $P''\equiv_{\alpha}P'\sigma$, $P\sigma\xrightarrow{\alpha\sigma}P''$.
\end{proposition}

\begin{proof}
By the definition of substitution (Definition \ref{subs}) and induction on the depth of inference.
\end{proof}

\begin{proposition}
\begin{enumerate}
  \item If $P\{w/z\}\xrightarrow{\alpha}P'$, where $w\notin fn(P)$ and $bn(\alpha)\cap fn(P,w)=\emptyset$, then there exist some $Q$ and $\beta$ with $Q\{w/z\}\equiv_{\alpha}P'$ and $\beta\sigma=\alpha$, $P\xrightarrow{\beta}Q$;
  \item If $P\{w/z\}\xrightarrow{\{\alpha_1,\cdots,\alpha_n\}}P'$, where $w\notin fn(P)$ and $bn(\alpha_1)\cap\cdots\cap bn(\alpha_n)\cap fn(P,w)=\emptyset$, then there exist some $Q$ and $\beta_1,\cdots,\beta_n$ with $Q\{w/z\}\equiv_{\alpha}P'$ and $\beta_1\sigma=\alpha_1,\cdots,\beta_n\sigma=\alpha_n$, $P\xrightarrow{\{\beta_1,\cdots,\beta_n\}}Q$.
\end{enumerate}

\end{proposition}

\begin{proof}
By the definition of substitution (Definition \ref{subs}) and induction on the depth of inference.
\end{proof} 

\section{Strongly Truly Concurrent Bisimilarities}\label{stcb}

\subsection{Basic Definitions}\label{STCC}

Firstly, in this subsection, we introduce concepts of (strongly) truly concurrent bisimilarities, including pomset bisimilarity, step bisimilarity, history-preserving (hp-)bisimilarity and hereditary history-preserving (hhp-)bisimilarity. In contrast to traditional truly concurrent bisimilarities in CTC \cite{CTC} and APTC \cite{ATC}, these versions in $\pi_{tc}$ must take care of actions with bound objects. Note that, these truly concurrent bisimilarities are defined as late bisimilarities, but not early bisimilarities, as defined in $\pi$-calculus \cite{PI1} \cite{PI2}. Note that, here, a PES $\mathcal{E}$ is deemed as a process.

\begin{definition}[Pomset transitions and step]
Let $\mathcal{E}$ be a PES and let $C\in\mathcal{C}(\mathcal{E})$, and $\emptyset\neq X\subseteq \mathbb{E}$, if $C\cap X=\emptyset$ and $C'=C\cup X\in\mathcal{C}(\mathcal{E})$, then $C\xrightarrow{X} C'$ is called a pomset transition from $C$ to $C'$. When the events in $X$ are pairwise concurrent, we say that $C\xrightarrow{X}C'$ is a step.
\end{definition}

\begin{definition}[Strong pomset, step bisimilarity]\label{PSB}
Let $\mathcal{E}_1$, $\mathcal{E}_2$ be PESs. A strong pomset bisimulation is a relation $R\subseteq\mathcal{C}(\mathcal{E}_1)\times\mathcal{C}(\mathcal{E}_2)$, such that if $(C_1,C_2)\in R$, and $C_1\xrightarrow{X_1}C_1'$ (with $\mathcal{E}_1\xrightarrow{X_1}\mathcal{E}_1'$) then $C_2\xrightarrow{X_2}C_2'$ (with $\mathcal{E}_2\xrightarrow{X_2}\mathcal{E}_2'$), with $X_1\subseteq \mathbb{E}_1$, $X_2\subseteq \mathbb{E}_2$, $X_1\sim X_2$ and $(C_1',C_2')\in R$:
\begin{enumerate}
  \item for each fresh action $\alpha\in X_1$, if $C_1''\xrightarrow{\alpha}C_1'''$ (with $\mathcal{E}_1''\xrightarrow{\alpha}\mathcal{E}_1'''$), then for some $C_2''$ and $C_2'''$, $C_2''\xrightarrow{\alpha}C_2'''$ (with $\mathcal{E}_2''\xrightarrow{\alpha}\mathcal{E}_2'''$), such that if $(C_1'',C_2'')\in R$ then $(C_1''',C_2''')\in R$;
  \item for each $x(y)\in X_1$ with ($y\notin n(\mathcal{E}_1, \mathcal{E}_2)$), if $C_1''\xrightarrow{x(y)}C_1'''$ (with $\mathcal{E}_1''\xrightarrow{x(y)}\mathcal{E}_1'''\{w/y\}$) for all $w$, then for some $C_2''$ and $C_2'''$, $C_2''\xrightarrow{x(y)}C_2'''$ (with $\mathcal{E}_2''\xrightarrow{x(y)}\mathcal{E}_2'''\{w/y\}$) for all $w$, such that if $(C_1'',C_2'')\in R$ then $(C_1''',C_2''')\in R$;
    \item for each two $x_1(y),x_2(y)\in X_1$ with ($y\notin n(\mathcal{E}_1, \mathcal{E}_2)$), if $C_1''\xrightarrow{\{x_1(y),x_2(y)\}}C_1'''$ (with $\mathcal{E}_1''\xrightarrow{\{x_1(y),x_2(y)\}}\mathcal{E}_1'''\{w/y\}$) for all $w$, then for some $C_2''$ and $C_2'''$, $C_2''\xrightarrow{\{x_1(y),x_2(y)\}}C_2'''$ (with $\mathcal{E}_2''\xrightarrow{\{x_1(y),x_2(y)\}}\mathcal{E}_2'''\{w/y\}$) for all $w$, such that if $(C_1'',C_2'')\in R$ then $(C_1''',C_2''')\in R$;
  \item for each $\overline{x}(y)\in X_1$ with $y\notin n(\mathcal{E}_1, \mathcal{E}_2)$, if $C_1''\xrightarrow{\overline{x}(y)}C_1'''$ (with $\mathcal{E}_1''\xrightarrow{\overline{x}(y)}\mathcal{E}_1'''$), then for some $C_2''$ and $C_2'''$, $C_2''\xrightarrow{\overline{x}(y)}C_2'''$ (with $\mathcal{E}_2''\xrightarrow{\overline{x}(y)}\mathcal{E}_2'''$), such that if $(C_1'',C_2'')\in R$ then $(C_1''',C_2''')\in R$.
\end{enumerate}
 and vice-versa.

We say that $\mathcal{E}_1$, $\mathcal{E}_2$ are strong pomset bisimilar, written $\mathcal{E}_1\sim_p\mathcal{E}_2$, if there exists a strong pomset bisimulation $R$, such that $(\emptyset,\emptyset)\in R$. By replacing pomset transitions with steps, we can get the definition of strong step bisimulation. When PESs $\mathcal{E}_1$ and $\mathcal{E}_2$ are strong step bisimilar, we write $\mathcal{E}_1\sim_s\mathcal{E}_2$.

%Note that, we do not admit multiple input with the same input variable in concurrency, that is, $\xrightarrow{\{x_1(y),x_2(y)\}}$ is forbidden. We assume that in each $x(y)\in X_1$, all $y$ are distinct, or all $x_1(y), x_2(y)\in X_1$ are in causality. For all $w_1$ and $w_2$, $P\nrightarrow^{\{x_1(y),x_2(y)\}}P'\{w_1/y\}\{w_2/y\}$, makes $P'\{w_1/y\}\{w_2/y\}$ do not occur. Actually, if we assume that all concurrent actions are put into parallel explicitly, then the side condition of the transition rule $\textbf{PAR}_3$ in Table \ref{TRForPITC} ($bn(\alpha)\cap fn(Q)=\emptyset,bn(\beta)\cap fn(P)=\emptyset$) also ensures that $\nrightarrow^{\{x_1(y),x_2(y)\}}$. This principle is throughout this paper, and we do not mention any more.
\end{definition}

\begin{definition}[Posetal product]
Given two PESs $\mathcal{E}_1$, $\mathcal{E}_2$, the posetal product of their configurations, denoted $\mathcal{C}(\mathcal{E}_1)\overline{\times}\mathcal{C}(\mathcal{E}_2)$, is defined as

$$\{(C_1,f,C_2)|C_1\in\mathcal{C}(\mathcal{E}_1),C_2\in\mathcal{C}(\mathcal{E}_2),f:C_1\rightarrow C_2 \textrm{ isomorphism}\}.$$

A subset $R\subseteq\mathcal{C}(\mathcal{E}_1)\overline{\times}\mathcal{C}(\mathcal{E}_2)$ is called a posetal relation. We say that $R$ is downward closed when for any $(C_1,f,C_2),(C_1',f',C_2')\in \mathcal{C}(\mathcal{E}_1)\overline{\times}\mathcal{C}(\mathcal{E}_2)$, if $(C_1,f,C_2)\subseteq (C_1',f',C_2')$ pointwise and $(C_1',f',C_2')\in R$, then $(C_1,f,C_2)\in R$.

For $f:X_1\rightarrow X_2$, we define $f[x_1\mapsto x_2]:X_1\cup\{x_1\}\rightarrow X_2\cup\{x_2\}$, $z\in X_1\cup\{x_1\}$,(1)$f[x_1\mapsto x_2](z)=
x_2$,if $z=x_1$;(2)$f[x_1\mapsto x_2](z)=f(z)$, otherwise. Where $X_1\subseteq \mathbb{E}_1$, $X_2\subseteq \mathbb{E}_2$, $x_1\in \mathbb{E}_1$, $x_2\in \mathbb{E}_2$.
\end{definition}

\begin{definition}[Strong (hereditary) history-preserving bisimilarity]\label{HHPB}
A strong history-preserving (hp-) bisimulation is a posetal relation $R\subseteq\mathcal{C}(\mathcal{E}_1)\overline{\times}\mathcal{C}(\mathcal{E}_2)$ such that if $(C_1,f,C_2)\in R$, and
\begin{enumerate}
  \item for $e_1=\alpha$ a fresh action, if $C_1\xrightarrow{\alpha}C_1'$ (with $\mathcal{E}_1\xrightarrow{\alpha}\mathcal{E}_1'$), then for some $C_2'$ and $e_2=\alpha$, $C_2\xrightarrow{\alpha}C_2'$ (with $\mathcal{E}_2\xrightarrow{\alpha}\mathcal{E}_2'$), such that $(C_1',f[e_1\mapsto e_2],C_2')\in R$;
  \item for $e_1=x(y)$ with ($y\notin n(\mathcal{E}_1, \mathcal{E}_2)$), if $C_1\xrightarrow{x(y)}C_1'$ (with $\mathcal{E}_1\xrightarrow{x(y)}\mathcal{E}_1'\{w/y\}$) for all $w$, then for some $C_2'$ and $e_2=x(y)$, $C_2\xrightarrow{x(y)}C_2'$ (with $\mathcal{E}_2\xrightarrow{x(y)}\mathcal{E}_2'\{w/y\}$) for all $w$, such that $(C_1',f[e_1\mapsto e_2],C_2')\in R$;
  \item for $e_1=\overline{x}(y)$ with $y\notin n(\mathcal{E}_1, \mathcal{E}_2)$, if $C_1\xrightarrow{\overline{x}(y)}C_1'$ (with $\mathcal{E}_1\xrightarrow{\overline{x}(y)}\mathcal{E}_1'$), then for some $C_2'$ and $e_2=\overline{x}(y)$, $C_2\xrightarrow{\overline{x}(y)}C_2'$ (with $\mathcal{E}_2\xrightarrow{\overline{x}(y)}\mathcal{E}_2'$), such that $(C_1',f[e_1\mapsto e_2],C_2')\in R$.
\end{enumerate}

and vice-versa. $\mathcal{E}_1,\mathcal{E}_2$ are strong history-preserving (hp-)bisimilar and are written $\mathcal{E}_1\sim_{hp}\mathcal{E}_2$ if there exists a strong hp-bisimulation $R$ such that $(\emptyset,\emptyset,\emptyset)\in R$.

A strongly hereditary history-preserving (hhp-)bisimulation is a downward closed strong hp-bisimulation. $\mathcal{E}_1,\mathcal{E}_2$ are strongly hereditary history-preserving (hhp-)bisimilar and are written $\mathcal{E}_1\sim_{hhp}\mathcal{E}_2$.
\end{definition}

Since the Parallel composition $\parallel$ is a fundamental computational pattern in CTC and APTC, and also it is fundamental in $\pi_{tc}$ as defined in Table \ref{TRForPITC}, and cannot be instead of other operators. So, the above truly concurrent bisimilarities are preserved by substitutions of names as defined in Definition \ref{subs}. We illustrate it by an example. We assume $P\equiv \overline{x}v$, abbreviated to $\overline{x}$; and $Q\equiv y(u)$, abbreviated to $y$. Then the following equations are true when $x\neq y$ and $u\neq v$:

$$P\parallel Q\sim_{p} \overline{x}\parallel y$$

$$P\parallel Q\sim_{s} \overline{x}\parallel y$$

$$P\parallel Q\sim_{hp} \overline{x}\parallel y$$

$$P\parallel Q\sim_{hhp} \overline{x}\parallel y.$$

By substituting $y$ to $x$, the following equations still hold:

$$P\parallel Q\{x/y\} \sim_{p} \overline{x}\parallel x$$

$$P\parallel Q\{x/y\} \sim_{s} \overline{x}\parallel x$$

$$P\parallel Q\{x/y\} \sim_{hp} \overline{x}\parallel x$$

$$P\parallel Q\{x/y\} \sim_{hhp} \overline{x}\parallel x.$$

\begin{theorem}
$\equiv_{\alpha}$ are strongly truly concurrent bisimulations. That is, if $P\equiv_{\alpha}Q$, then,
\begin{enumerate}
  \item $P\sim_p Q$;
  \item $P\sim_s Q$;
  \item $P\sim_{hp} Q$;
  \item $P\sim_{hhp} Q$.
\end{enumerate}
\end{theorem}

\begin{proof}
By induction on the depth of inference (see Table \ref{TRForPITC}), we can get the following facts:

\begin{enumerate}
  \item If $\alpha$ is a free action and $P\xrightarrow{\alpha}P'$, then equally for some $Q'$ with $P'\equiv_{\alpha}Q'$, $Q\xrightarrow{\alpha}Q'$;
  \item If $P\xrightarrow{a(y)}P'$ with $a=x$ or $a=\overline{x}$ and $z\notin n(Q)$, then equally for some $Q'$ with $P'\{z/y\}\equiv_{\alpha}Q'$, $Q\xrightarrow{a(z)}Q'$.
\end{enumerate}

Then, we can get:

\begin{enumerate}
  \item by the definition of strong pomset bisimilarity (Definition \ref{PSB}), $P\sim_p Q$;
  \item by the definition of strong step bisimilarity (Definition \ref{PSB}), $P\sim_s Q$;
  \item by the definition of strong hp-bisimilarity (Definition \ref{HHPB}), $P\sim_{hp} Q$;
  \item by the definition of strongly hhp-bisimilarity (Definition \ref{HHPB}), $P\sim_{hhp} Q$.
\end{enumerate}
\end{proof}

\subsection{Laws and Congruence}

Similarly to CTC \cite{CTC}, we can obtain the following laws with respect to truly concurrent bisimilarities.

\begin{theorem}[Summation laws for strong pomset bisimilarity]
The summation laws for strong pomset bisimilarity are as follows.
\begin{enumerate}
  \item $P+\textbf{nil}\sim_p P$;
  \item $P+P\sim_p P$;
  \item $P_1+P_2\sim_p P_2+P_1$;
  \item $P_1+(P_2+P_3)\sim_p (P_1+P_2)+P_3$.
\end{enumerate}
\end{theorem}

\begin{proof}
\begin{enumerate}
  \item It is sufficient to prove the relation $R=\{(P+\textbf{nil}, P)\}\cup \textbf{Id}$ is a strong pomset bisimulation. It can be proved similarly to the proof of Monoid laws for strong pomset bisimulation in CTC \cite{CTC}, we omit it;
  \item It is sufficient to prove the relation $R=\{(P+P, P)\}\cup \textbf{Id}$ is a strong pomset bisimulation. It can be proved similarly to the proof of Monoid laws for strong pomset bisimulation in CTC \cite{CTC}, we omit it;
  \item It is sufficient to prove the relation $R=\{(P_1+P_2, P_2+P_1)\}\cup \textbf{Id}$ is a strong pomset bisimulation. It can be proved similarly to the proof of Monoid laws for strong pomset bisimulation in CTC \cite{CTC}, we omit it;
  \item It is sufficient to prove the relation $R=\{(P_1+(P_2+P_3), (P_1+P_2)+P_3)\}\cup \textbf{Id}$ is a strong pomset bisimulation. It can be proved similarly to the proof of Monoid laws for strong pomset bisimulation in CTC \cite{CTC}, we omit it.
\end{enumerate}
\end{proof}

\begin{theorem}[Summation laws for strong step bisimilarity]
The summation laws for strong step bisimilarity are as follows.
\begin{enumerate}
  \item $P+\textbf{nil}\sim_s P$;
  \item $P+P\sim_s P$;
  \item $P_1+P_2\sim_s P_2+P_1$;
  \item $P_1+(P_2+P_3)\sim_s (P_1+P_2)+P_3$.
\end{enumerate}
\end{theorem}

\begin{proof}
\begin{enumerate}
  \item It is sufficient to prove the relation $R=\{(P+\textbf{nil}, P)\}\cup \textbf{Id}$ is a strong step bisimulation. It can be proved similarly to the proof of Monoid laws for strong step bisimulation in CTC \cite{CTC}, we omit it;
  \item It is sufficient to prove the relation $R=\{(P+P, P)\}\cup \textbf{Id}$ is a strong step bisimulation. It can be proved similarly to the proof of Monoid laws for strong step bisimulation in CTC \cite{CTC}, we omit it;
  \item It is sufficient to prove the relation $R=\{(P_1+P_2, P_2+P_1)\}\cup \textbf{Id}$ is a strong step bisimulation. It can be proved similarly to the proof of Monoid laws for strong step bisimulation in CTC \cite{CTC}, we omit it;
  \item It is sufficient to prove the relation $R=\{(P_1+(P_2+P_3), (P_1+P_2)+P_3)\}\cup \textbf{Id}$ is a strong step bisimulation. It can be proved similarly to the proof of Monoid laws for strong step bisimulation in CTC \cite{CTC}, we omit it.
\end{enumerate}
\end{proof}

\begin{theorem}[Summation laws for strong hp-bisimilarity]
The summation laws for strong hp-bisimilarity are as follows.
\begin{enumerate}
  \item $P+\textbf{nil}\sim_{hp} P$;
  \item $P+P\sim_{hp} P$;
  \item $P_1+P_2\sim_{hp} P_2+P_1$;
  \item $P_1+(P_2+P_3)\sim_{hp} (P_1+P_2)+P_3$.
\end{enumerate}
\end{theorem}

\begin{proof}
\begin{enumerate}
  \item It is sufficient to prove the relation $R=\{(P+\textbf{nil}, P)\}\cup \textbf{Id}$ is a strong hp-bisimulation. It can be proved similarly to the proof of Monoid laws for strong hp-bisimulation in CTC \cite{CTC}, we omit it;
  \item It is sufficient to prove the relation $R=\{(P+P, P)\}\cup \textbf{Id}$ is a strong hp-bisimulation. It can be proved similarly to the proof of Monoid laws for strong hp-bisimulation in CTC \cite{CTC}, we omit it;
  \item It is sufficient to prove the relation $R=\{(P_1+P_2, P_2+P_1)\}\cup \textbf{Id}$ is a strong hp-bisimulation. It can be proved similarly to the proof of Monoid laws for strong hp-bisimulation in CTC \cite{CTC}, we omit it;
  \item It is sufficient to prove the relation $R=\{(P_1+(P_2+P_3), (P_1+P_2)+P_3)\}\cup \textbf{Id}$ is a strong hp-bisimulation. It can be proved similarly to the proof of Monoid laws for strong hp-bisimulation in CTC \cite{CTC}, we omit it.
\end{enumerate}
\end{proof}

\begin{theorem}[Summation laws for strongly hhp-bisimilarity]
The summation laws for strongly hhp-bisimilarity are as follows.
\begin{enumerate}
  \item $P+\textbf{nil}\sim_{hhp} P$;
  \item $P+P\sim_{hhp} P$;
  \item $P_1+P_2\sim_{hhp} P_2+P_1$;
  \item $P_1+(P_2+P_3)\sim_{hhp} (P_1+P_2)+P_3$.
\end{enumerate}
\end{theorem}

\begin{proof}
\begin{enumerate}
  \item It is sufficient to prove the relation $R=\{(P+\textbf{nil}, P)\}\cup \textbf{Id}$ is a strongly hhp-bisimulation. It can be proved similarly to the proof of Monoid laws for strongly hhp-bisimulation in CTC \cite{CTC}, we omit it;
  \item It is sufficient to prove the relation $R=\{(P+P, P)\}\cup \textbf{Id}$ is a strongly hhp-bisimulation. It can be proved similarly to the proof of Monoid laws for strongly hhp-bisimulation in CTC \cite{CTC}, we omit it;
  \item It is sufficient to prove the relation $R=\{(P_1+P_2, P_2+P_1)\}\cup \textbf{Id}$ is a strongly hhp-bisimulation. It can be proved similarly to the proof of Monoid laws for strongly hhp-bisimulation in CTC \cite{CTC}, we omit it;
  \item It is sufficient to prove the relation $R=\{(P_1+(P_2+P_3), (P_1+P_2)+P_3)\}\cup \textbf{Id}$ is a strongly hhp-bisimulation. It can be proved similarly to the proof of Monoid laws for strongly hhp-bisimulation in CTC \cite{CTC}, we omit it.
\end{enumerate}
\end{proof}

\begin{theorem}[Identity law for truly concurrent bisimilarities]
If $A(\widetilde{x})\overset{\text{def}}{=}P$, then
\begin{enumerate}
  \item $A(\widetilde{y})\sim_p P\{\widetilde{y}/\widetilde{x}\}$;
  \item $A(\widetilde{y})\sim_s P\{\widetilde{y}/\widetilde{x}\}$;
  \item $A(\widetilde{y})\sim_{hp} P\{\widetilde{y}/\widetilde{x}\}$;
  \item $A(\widetilde{y})\sim_{hhp} P\{\widetilde{y}/\widetilde{x}\}$.
\end{enumerate}
\end{theorem}

\begin{proof}
\begin{enumerate}
  \item It is straightforward to see that $R=\{A(\widetilde{y},P\{\widetilde{y}/\widetilde{x}\})\}\cup \textbf{Id}$ is a strong pomset bisimulation;
  \item It is straightforward to see that $R=\{A(\widetilde{y},P\{\widetilde{y}/\widetilde{x}\})\}\cup \textbf{Id}$ is a strong step bisimulation;
  \item It is straightforward to see that $R=\{A(\widetilde{y},P\{\widetilde{y}/\widetilde{x}\})\}\cup \textbf{Id}$ is a strong hp-bisimulation;
  \item It is straightforward to see that $R=\{A(\widetilde{y},P\{\widetilde{y}/\widetilde{x}\})\}\cup \textbf{Id}$ is a strongly hhp-bisimulation;
\end{enumerate}
\end{proof}

\begin{theorem}[Restriction Laws for strong pomset bisimilarity]
The restriction laws for strong pomset bisimilarity are as follows.

\begin{enumerate}
  \item $(y)P\sim_p P$, if $y\notin fn(P)$;
  \item $(y)(z)P\sim_p (z)(y)P$;
  \item $(y)(P+Q)\sim_p (y)P+(y)Q$;
  \item $(y)\alpha.P\sim_p \alpha.(y)P$ if $y\notin n(\alpha)$;
  \item $(y)\alpha.P\sim_p \textbf{nil}$ if $y$ is the subject of $\alpha$.
\end{enumerate}
\end{theorem}

\begin{proof}
\begin{enumerate}
  \item It is sufficient to prove the relation $R=\{((y)P, P)|\textrm{ if }y\notin fn(P)\}\cup \textbf{Id}$ is a strong pomset bisimulation. It can be proved similarly to the proof of Static laws about restriction $\setminus$ for strong pomset bisimulation in CTC \cite{CTC}, we omit it;
  \item It is sufficient to prove the relation $R=\{((y)(z)P, (z)(y)P)\}\cup \textbf{Id}$ is a strong pomset bisimulation. It can be proved similarly to the proof of Static laws about restriction $\setminus$ for strong pomset bisimulation in CTC \cite{CTC}, we omit it;
  \item It is sufficient to prove the relation $R=\{((y)(P+Q), (y)P+(y)Q)\}\cup \textbf{Id}$ is a strong pomset bisimulation. It can be proved similarly to the proof of Static laws about restriction $\setminus$ for strong pomset bisimulation in CTC \cite{CTC}, we omit it;
  \item It is sufficient to prove the relation $R=\{((y)\alpha.P, \alpha.(y)P)|\textrm{ if }y\notin n(\alpha)\}\cup \textbf{Id}$ is a strong pomset bisimulation. It can be proved similarly to the proof of Static laws about restriction $\setminus$ for strong pomset bisimulation in CTC \cite{CTC}, we omit it;
  \item It is sufficient to prove the relation $R=\{((y)\alpha.P, \textbf{nil})|\textrm{ if }y\textrm{ is the subject of }\alpha\}\cup \textbf{Id}$ is a strong pomset bisimulation. It can be proved similarly to the proof of Static laws about restriction $\setminus$ for strong pomset bisimulation in CTC \cite{CTC}, we omit it.
\end{enumerate}
\end{proof}

\begin{theorem}[Restriction Laws for strong step bisimilarity]
The restriction laws for strong step bisimilarity are as follows.

\begin{enumerate}
  \item $(y)P\sim_s P$, if $y\notin fn(P)$;
  \item $(y)(z)P\sim_s (z)(y)P$;
  \item $(y)(P+Q)\sim_s (y)P+(y)Q$;
  \item $(y)\alpha.P\sim_s \alpha.(y)P$ if $y\notin n(\alpha)$;
  \item $(y)\alpha.P\sim_s \textbf{nil}$ if $y$ is the subject of $\alpha$.
\end{enumerate}
\end{theorem}

\begin{proof}
\begin{enumerate}
  \item It is sufficient to prove the relation $R=\{((y)P, P)|\textrm{ if }y\notin fn(P)\}\cup \textbf{Id}$ is a strong step bisimulation. It can be proved similarly to the proof of Static laws about restriction $\setminus$ for strong step bisimulation in CTC \cite{CTC}, we omit it;
  \item It is sufficient to prove the relation $R=\{((y)(z)P, (z)(y)P)\}\cup \textbf{Id}$ is a strong step bisimulation. It can be proved similarly to the proof of Static laws about restriction $\setminus$ for strong step bisimulation in CTC \cite{CTC}, we omit it;
  \item It is sufficient to prove the relation $R=\{((y)(P+Q), (y)P+(y)Q)\}\cup \textbf{Id}$ is a strong step bisimulation. It can be proved similarly to the proof of Static laws about restriction $\setminus$ for strong step bisimulation in CTC \cite{CTC}, we omit it;
  \item It is sufficient to prove the relation $R=\{((y)\alpha.P, \alpha.(y)P)|\textrm{ if }y\notin n(\alpha)\}\cup \textbf{Id}$ is a strong step bisimulation. It can be proved similarly to the proof of Static laws about restriction $\setminus$ for strong step bisimulation in CTC \cite{CTC}, we omit it;
  \item It is sufficient to prove the relation $R=\{((y)\alpha.P, \textbf{nil})|\textrm{ if }y\textrm{ is the subject of }\alpha\}\cup \textbf{Id}$ is a strong step bisimulation. It can be proved similarly to the proof of Static laws about restriction $\setminus$ for strong step bisimulation in CTC \cite{CTC}, we omit it.
\end{enumerate}
\end{proof}

\begin{theorem}[Restriction Laws for strong hp-bisimilarity]
The restriction laws for strong hp-bisimilarity are as follows.

\begin{enumerate}
  \item $(y)P\sim_{hp} P$, if $y\notin fn(P)$;
  \item $(y)(z)P\sim_{hp} (z)(y)P$;
  \item $(y)(P+Q)\sim_{hp} (y)P+(y)Q$;
  \item $(y)\alpha.P\sim_{hp} \alpha.(y)P$ if $y\notin n(\alpha)$;
  \item $(y)\alpha.P\sim_{hp} \textbf{nil}$ if $y$ is the subject of $\alpha$.
\end{enumerate}
\end{theorem}

\begin{proof}
\begin{enumerate}
  \item It is sufficient to prove the relation $R=\{((y)P, P)|\textrm{ if }y\notin fn(P)\}\cup \textbf{Id}$ is a strong hp-bisimulation. It can be proved similarly to the proof of Static laws about restriction $\setminus$ for strong hp-bisimulation in CTC \cite{CTC}, we omit it;
  \item It is sufficient to prove the relation $R=\{((y)(z)P, (z)(y)P)\}\cup \textbf{Id}$ is a strong hp-bisimulation. It can be proved similarly to the proof of Static laws about restriction $\setminus$ for strong hp-bisimulation in CTC \cite{CTC}, we omit it;
  \item It is sufficient to prove the relation $R=\{((y)(P+Q), (y)P+(y)Q)\}\cup \textbf{Id}$ is a strong hp-bisimulation. It can be proved similarly to the proof of Static laws about restriction $\setminus$ for strong hp-bisimulation in CTC \cite{CTC}, we omit it;
  \item It is sufficient to prove the relation $R=\{((y)\alpha.P, \alpha.(y)P)|\textrm{ if }y\notin n(\alpha)\}\cup \textbf{Id}$ is a strong hp-bisimulation. It can be proved similarly to the proof of Static laws about restriction $\setminus$ for strong hp-bisimulation in CTC \cite{CTC}, we omit it;
  \item It is sufficient to prove the relation $R=\{((y)\alpha.P, \textbf{nil})|\textrm{ if }y\textrm{ is the subject of }\alpha\}\cup \textbf{Id}$ is a strong hp-bisimulation. It can be proved similarly to the proof of Static laws about restriction $\setminus$ for strong hp-bisimulation in CTC \cite{CTC}, we omit it.
\end{enumerate}
\end{proof}

\begin{theorem}[Restriction Laws for strongly hhp-bisimilarity]
The restriction laws for strongly hhp-bisimilarity are as follows.

\begin{enumerate}
  \item $(y)P\sim_{hhp} P$, if $y\notin fn(P)$;
  \item $(y)(z)P\sim_{hhp} (z)(y)P$;
  \item $(y)(P+Q)\sim_{hhp} (y)P+(y)Q$;
  \item $(y)\alpha.P\sim_{hhp} \alpha.(y)P$ if $y\notin n(\alpha)$;
  \item $(y)\alpha.P\sim_{hhp} \textbf{nil}$ if $y$ is the subject of $\alpha$.
\end{enumerate}
\end{theorem}

\begin{proof}
\begin{enumerate}
  \item It is sufficient to prove the relation $R=\{((y)P, P)|\textrm{ if }y\notin fn(P)\}\cup \textbf{Id}$ is a strongly hhp-bisimulation. It can be proved similarly to the proof of Static laws about restriction $\setminus$ for strongly hhp-bisimulation in CTC \cite{CTC}, we omit it;
  \item It is sufficient to prove the relation $R=\{((y)(z)P, (z)(y)P)\}\cup \textbf{Id}$ is a strongly hhp-bisimulation. It can be proved similarly to the proof of Static laws about restriction $\setminus$ for strongly hhp-bisimulation in CTC \cite{CTC}, we omit it;
  \item It is sufficient to prove the relation $R=\{((y)(P+Q), (y)P+(y)Q)\}\cup \textbf{Id}$ is a strongly hhp-bisimulation. It can be proved similarly to the proof of Static laws about restriction $\setminus$ for strongly hhp-bisimulation in CTC \cite{CTC}, we omit it;
  \item It is sufficient to prove the relation $R=\{((y)\alpha.P, \alpha.(y)P)|\textrm{ if }y\notin n(\alpha)\}\cup \textbf{Id}$ is a strongly hhp-bisimulation. It can be proved similarly to the proof of Static laws about restriction $\setminus$ for strongly hhp-bisimulation in CTC \cite{CTC}, we omit it;
  \item It is sufficient to prove the relation $R=\{((y)\alpha.P, \textbf{nil})|\textrm{ if }y\textrm{ is the subject of }\alpha\}\cup \textbf{Id}$ is a strongly hhp-bisimulation. It can be proved similarly to the proof of Static laws about restriction $\setminus$ for strongly hhp-bisimulation in CTC \cite{CTC}, we omit it.
\end{enumerate}
\end{proof}

\begin{theorem}[Parallel laws for strong pomset bisimilarity]
The parallel laws for strong pomset bisimilarity are as follows.
\begin{enumerate}
  \item $P\parallel \textbf{nil}\sim_p P$;
  \item $P_1\parallel P_2\sim_p P_2\parallel P_1$;
  \item $(y)P_1\parallel P_2\sim_p (y)(P_1\parallel P_2)$
  \item $(P_1\parallel P_2)\parallel P_3\sim_p P_1\parallel (P_2\parallel P_3)$;
  \item $(y)(P_1\parallel P_2)\sim_p (y)P_1\parallel (y)P_2$, if $y\notin fn(P_1)\cap fn(P_2)$.
\end{enumerate}
\end{theorem}

\begin{proof}
\begin{enumerate}
  \item It is sufficient to prove the relation $R=\{(P\parallel \textbf{nil}, P)\}\cup \textbf{Id}$ is a strong pomset bisimulation. It can be proved similarly to the proof of Static laws about parallel $\parallel$ for strong pomset bisimulation in CTC \cite{CTC}, we omit it;
  \item It is sufficient to prove the relation $R=\{(P_1\parallel P_2, P_2\parallel P_1)\}\cup \textbf{Id}$ is a strong pomset bisimulation. It can be proved similarly to the proof of Static laws about parallel $\parallel$ for strong pomset bisimulation in CTC \cite{CTC}, we omit it;
  \item It is sufficient to prove the relation $R=\{((y)P_1\parallel P_2, (y)(P_1\parallel P_2))\}\cup \textbf{Id}$ is a strong pomset bisimulation. It can be proved similarly to the proof of Static laws about parallel $\parallel$ for strong pomset bisimulation in CTC \cite{CTC}, we omit it;
  \item It is sufficient to prove the relation $R=\{((P_1\parallel P_2)\parallel P_3, P_1\parallel (P_2\parallel P_3))\}\cup \textbf{Id}$ is a strong pomset bisimulation. It can be proved similarly to the proof of Static laws about parallel $\parallel$ for strong pomset bisimulation in CTC \cite{CTC}, we omit it;
  \item It is sufficient to prove the relation $R=\{(y)(P_1\parallel P_2), (y)P_1\parallel (y)P_2)|\textrm{ if }y\notin fn(P_1)\cap fn(P_2)\}\cup \textbf{Id}$ is a strong pomset bisimulation. It can be proved similarly to the proof of Static laws about parallel $\parallel$ for strong pomset bisimulation in CTC \cite{CTC}, we omit it.
\end{enumerate}
\end{proof}

\begin{theorem}[Parallel laws for strong step bisimilarity]
The parallel laws for strong step bisimilarity are as follows.
\begin{enumerate}
  \item $P\parallel \textbf{nil}\sim_s P$;
  \item $P_1\parallel P_2\sim_s P_2\parallel P_1$;
  \item $(y)P_1\parallel P_2\sim_s (y)(P_1\parallel P_2)$
  \item $(P_1\parallel P_2)\parallel P_3\sim_s P_1\parallel (P_2\parallel P_3)$;
  \item $(y)(P_1\parallel P_2)\sim_s (y)P_1\parallel (y)P_2$, if $y\notin fn(P_1)\cap fn(P_2)$.
\end{enumerate}
\end{theorem}

\begin{proof}
\begin{enumerate}
  \item It is sufficient to prove the relation $R=\{(P\parallel \textbf{nil}, P)\}\cup \textbf{Id}$ is a strong step bisimulation. It can be proved similarly to the proof of Static laws about parallel $\parallel$ for strong step bisimulation in CTC \cite{CTC}, we omit it;
  \item It is sufficient to prove the relation $R=\{(P_1\parallel P_2, P_2\parallel P_1)\}\cup \textbf{Id}$ is a strong step bisimulation. It can be proved similarly to the proof of Static laws about parallel $\parallel$ for strong step bisimulation in CTC \cite{CTC}, we omit it;
  \item It is sufficient to prove the relation $R=\{((y)P_1\parallel P_2, (y)(P_1\parallel P_2))\}\cup \textbf{Id}$ is a strong step bisimulation. It can be proved similarly to the proof of Static laws about parallel $\parallel$ for strong step bisimulation in CTC \cite{CTC}, we omit it;
  \item It is sufficient to prove the relation $R=\{((P_1\parallel P_2)\parallel P_3, P_1\parallel (P_2\parallel P_3))\}\cup \textbf{Id}$ is a strong step bisimulation. It can be proved similarly to the proof of Static laws about parallel $\parallel$ for strong step bisimulation in CTC \cite{CTC}, we omit it;
  \item It is sufficient to prove the relation $R=\{(y)(P_1\parallel P_2), (y)P_1\parallel (y)P_2)|\textrm{ if }y\notin fn(P_1)\cap fn(P_2)\}\cup \textbf{Id}$ is a strong step bisimulation. It can be proved similarly to the proof of Static laws about parallel $\parallel$ for strong step bisimulation in CTC \cite{CTC}, we omit it.
\end{enumerate}
\end{proof}

\begin{theorem}[Parallel laws for strong hp-bisimilarity]
The parallel laws for strong hp-bisimilarity are as follows.
\begin{enumerate}
  \item $P\parallel \textbf{nil}\sim_{hp} P$;
  \item $P_1\parallel P_2\sim_{hp} P_2\parallel P_1$;
  \item $(y)P_1\parallel P_2\sim_{hp} (y)(P_1\parallel P_2)$
  \item $(P_1\parallel P_2)\parallel P_3\sim_{hp} P_1\parallel (P_2\parallel P_3)$;
  \item $(y)(P_1\parallel P_2)\sim_{hp} (y)P_1\parallel (y)P_2$, if $y\notin fn(P_1)\cap fn(P_2)$.
\end{enumerate}
\end{theorem}

\begin{proof}
\begin{enumerate}
  \item It is sufficient to prove the relation $R=\{(P\parallel \textbf{nil}, P)\}\cup \textbf{Id}$ is a strong hp-bisimulation. It can be proved similarly to the proof of Static laws about parallel $\parallel$ for strong hp-bisimulation in CTC \cite{CTC}, we omit it;
  \item It is sufficient to prove the relation $R=\{(P_1\parallel P_2, P_2\parallel P_1)\}\cup \textbf{Id}$ is a strong hp-bisimulation. It can be proved similarly to the proof of Static laws about parallel $\parallel$ for strong hp-bisimulation in CTC \cite{CTC}, we omit it;
  \item It is sufficient to prove the relation $R=\{((y)P_1\parallel P_2, (y)(P_1\parallel P_2))\}\cup \textbf{Id}$ is a strong hp-bisimulation. It can be proved similarly to the proof of Static laws about parallel $\parallel$ for strong hp-bisimulation in CTC \cite{CTC}, we omit it;
  \item It is sufficient to prove the relation $R=\{((P_1\parallel P_2)\parallel P_3, P_1\parallel (P_2\parallel P_3))\}\cup \textbf{Id}$ is a strong hp-bisimulation. It can be proved similarly to the proof of Static laws about parallel $\parallel$ for strong hp-bisimulation in CTC \cite{CTC}, we omit it;
  \item It is sufficient to prove the relation $R=\{(y)(P_1\parallel P_2), (y)P_1\parallel (y)P_2)|\textrm{ if }y\notin fn(P_1)\cap fn(P_2)\}\cup \textbf{Id}$ is a strong hp-bisimulation. It can be proved similarly to the proof of Static laws about parallel $\parallel$ for strong hp-bisimulation in CTC \cite{CTC}, we omit it.
\end{enumerate}
\end{proof}

\begin{theorem}[Parallel laws for strongly hhp-bisimilarity]
The parallel laws for strongly hhp-bisimilarity are as follows.
\begin{enumerate}
  \item $P\parallel \textbf{nil}\sim_{hhp} P$;
  \item $P_1\parallel P_2\sim_{hhp} P_2\parallel P_1$;
  \item $(y)P_1\parallel P_2\sim_{hhp} (y)(P_1\parallel P_2)$
  \item $(P_1\parallel P_2)\parallel P_3\sim_{hhp} P_1\parallel (P_2\parallel P_3)$;
  \item $(y)(P_1\parallel P_2)\sim_{hhp} (y)P_1\parallel (y)P_2$, if $y\notin fn(P_1)\cap fn(P_2)$.
\end{enumerate}
\end{theorem}

\begin{proof}
\begin{enumerate}
  \item It is sufficient to prove the relation $R=\{(P\parallel \textbf{nil}, P)\}\cup \textbf{Id}$ is a strongly hhp-bisimulation. It can be proved similarly to the proof of Static laws about parallel $\parallel$ for strongly hhp-bisimulation in CTC \cite{CTC}, we omit it;
  \item It is sufficient to prove the relation $R=\{(P_1\parallel P_2, P_2\parallel P_1)\}\cup \textbf{Id}$ is a strongly hhp-bisimulation. It can be proved similarly to the proof of Static laws about parallel $\parallel$ for strongly hhp-bisimulation in CTC \cite{CTC}, we omit it;
  \item It is sufficient to prove the relation $R=\{((y)P_1\parallel P_2, (y)(P_1\parallel P_2))\}\cup \textbf{Id}$ is a strongly hhp-bisimulation. It can be proved similarly to the proof of Static laws about parallel $\parallel$ for strongly hhp-bisimulation in CTC \cite{CTC}, we omit it;
  \item It is sufficient to prove the relation $R=\{((P_1\parallel P_2)\parallel P_3, P_1\parallel (P_2\parallel P_3))\}\cup \textbf{Id}$ is a strongly hhp-bisimulation. It can be proved similarly to the proof of Static laws about parallel $\parallel$ for strongly hhp-bisimulation in CTC \cite{CTC}, we omit it;
  \item It is sufficient to prove the relation $R=\{(y)(P_1\parallel P_2), (y)P_1\parallel (y)P_2)|\textrm{ if }y\notin fn(P_1)\cap fn(P_2)\}\cup \textbf{Id}$ is a strongly hhp-bisimulation. It can be proved similarly to the proof of Static laws about parallel $\parallel$ for strongly hhp-bisimulation in CTC \cite{CTC}, we omit it.
\end{enumerate}
\end{proof}

\begin{theorem}[Expansion law for truly concurrent bisimilarities]
Let $P\equiv\sum_i \alpha_i.P_i$ and $Q\equiv\sum_j\beta_j.Q_j$, where $bn(\alpha_i)\cap fn(Q)=\emptyset$ for all $i$, and $bn(\beta_j)\cap fn(P)=\emptyset$ for all $j$. Then

\begin{enumerate}
  \item $P\parallel Q\sim_p \sum_i\sum_j (\alpha_i\parallel \beta_j).(P_i\parallel Q_j)+\sum_{\alpha_i \textrm{ comp }\beta_j}\tau.R_{ij}$;
  \item $P\parallel Q\sim_s \sum_i\sum_j (\alpha_i\parallel \beta_j).(P_i\parallel Q_j)+\sum_{\alpha_i \textrm{ comp }\beta_j}\tau.R_{ij}$;
  \item $P\parallel Q\sim_{hp} \sum_i\sum_j (\alpha_i\parallel \beta_j).(P_i\parallel Q_j)+\sum_{\alpha_i \textrm{ comp }\beta_j}\tau.R_{ij}$;
  \item $P\parallel Q\nsim_{hhp} \sum_i\sum_j (\alpha_i\parallel \beta_j).(P_i\parallel Q_j)+\sum_{\alpha_i \textrm{ comp }\beta_j}\tau.R_{ij}$.
\end{enumerate}

Where $\alpha_i$ comp $\beta_j$ and $R_{ij}$ are defined as follows:
\begin{enumerate}
  \item $\alpha_i$ is $\overline{x}u$ and $\beta_j$ is $x(v)$, then $R_{ij}=P_i\parallel Q_j\{u/v\}$;
  \item $\alpha_i$ is $\overline{x}(u)$ and $\beta_j$ is $x(v)$, then $R_{ij}=(w)(P_i\{w/u\}\parallel Q_j\{w/v\})$, if $w\notin fn((u)P_i)\cup fn((v)Q_j)$;
  \item $\alpha_i$ is $x(v)$ and $\beta_j$ is $\overline{x}u$, then $R_{ij}=P_i\{u/v\}\parallel Q_j$;
  \item $\alpha_i$ is $x(v)$ and $\beta_j$ is $\overline{x}(u)$, then $R_{ij}=(w)(P_i\{w/v\}\parallel Q_j\{w/u\})$, if $w\notin fn((v)P_i)\cup fn((u)Q_j)$.
\end{enumerate}
\end{theorem}

\begin{proof}
\begin{enumerate}
  \item It is sufficient to prove the relation $R=\{(P\parallel Q, \sum_i\sum_j (\alpha_i\parallel \beta_j).(P_i\parallel Q_j)+\sum_{\alpha_i \textrm{ comp }\beta_j}\tau.R_{ij})|\textrm{ if }y\notin fn(P)\}\cup \textbf{Id}$ is a strong pomset bisimulation. It can be proved similarly to the proof of Expansion law for strong pomset bisimulation in CTC \cite{CTC}, we omit it;
  \item It is sufficient to prove the relation $R=\{(P\parallel Q, \sum_i\sum_j (\alpha_i\parallel \beta_j).(P_i\parallel Q_j)+\sum_{\alpha_i \textrm{ comp }\beta_j}\tau.R_{ij})|\textrm{ if }y\notin fn(P)\}\cup \textbf{Id}$ is a strong step bisimulation. It can be proved similarly to the proof of Expansion law for strong step bisimulation in CTC \cite{CTC}, we omit it;
  \item It is sufficient to prove the relation $R=\{(P\parallel Q, \sum_i\sum_j (\alpha_i\parallel \beta_j).(P_i\parallel Q_j)+\sum_{\alpha_i \textrm{ comp }\beta_j}\tau.R_{ij})|\textrm{ if }y\notin fn(P)\}\cup \textbf{Id}$ is a strong hp-bisimulation. It can be proved similarly to the proof of Expansion law for strong hp-bisimulation in CTC \cite{CTC}, we omit it;
  \item We just prove that for free actions $a,b,c$, let $s_1=(a+ b)\parallel c$, $t_1=(a\parallel c)+ (b\parallel c)$, and $s_2=a\parallel (b+ c)$, $t_2=(a\parallel b)+ (a\parallel c)$. We know that $s_1\sim_{hp} t_1$ and $s_2\sim_{hp} t_2$, we prove that $s_1\nsim_{hhp} t_1$ and $s_2\nsim_{hhp} t_2$. Let $(C(s_1),f_1,C(t_1))$ and $(C(s_2),f_2,C(t_2))$ are the corresponding posetal products.
    \begin{itemize}
        \item $s_1\nsim_{hhp} t_1$. $s_1\xrightarrow{\{a,c\}}\surd(s_1')$ ($C(s_1)\xrightarrow{\{a,c\}}C(s_1')$), then $t_1\xrightarrow{\{a,c\}}\surd(t_1')$ ($C(t_1)\xrightarrow{\{a,c\}}C(t_1')$), we define $f_1'=f_1[a\mapsto a, c\mapsto c]$, obviously, $(C(s_1),f_1,C(t_1))\in \sim_{hp}$ and $(C(s_1'),f_1',C(t_1'))\in \sim_{hp}$. But, $(C(s_1),f_1,C(t_1))\in \sim_{hhp}$ and $(C(s_1'),f_1',C(t_1'))\in \nsim_{hhp}$, just because they are not downward closed. Let $(C(s_1''),f_1'',C(t_1''))$, and $f_1''=f_1[c\mapsto c]$, $s_1\xrightarrow{c}s_1''$ ($C(s_1)\xrightarrow{c}C(s_1'')$), $t_1\xrightarrow{c}t_1''$ ($C(t_1)\xrightarrow{c}C(t_1'')$), it is easy to see that $(C(s_1''),f_1'',C(t_1''))\subseteq (C(s_1'),f_1',C(t_1'))$ pointwise, while $(C(s_1''),f_1'',C(t_1''))\notin \sim_{hp}$, because $s_1''$ and $C(s_1'')$ exist, but $t_1''$ and $C(t_1'')$ do not exist.
        \item $s_2\nsim_{hhp} t_2$. $s_2\xrightarrow{\{a,c\}}\surd(s_2')$ ($C(s_2)\xrightarrow{\{a,c\}}C(s_2')$), then $t_2\xrightarrow{\{a,c\}}\surd(t_2')$ ($C(t_2)\xrightarrow{\{a,c\}}C(t_2')$), we define $f_2'=f_2[a\mapsto a, c\mapsto c]$, obviously, $(C(s_2),f_2,C(t_2))\in \sim_{hp}$ and $(C(s_2'),f_2',C(t_2'))\in \sim_{hp}$. But, $(C(s_2),f_2,C(t_2))\in \sim_{hhp}$ and $(C(s_2'),f_2',C(t_2'))\in \nsim_{hhp}$, just because they are not downward closed. Let $(C(s_2''),f_2'',C(t_2''))$, and $f_2''=f_2[a\mapsto a]$, $s_2\xrightarrow{a}s_2''$ ($C(s_2)\xrightarrow{a}C(s_2'')$), $t_2\xrightarrow{a}t_2''$ ($C(t_2)\xrightarrow{a}C(t_2'')$), it is easy to see that $(C(s_2''),f_2'',C(t_2''))\subseteq (C(s_2'),f_2',C(t_2'))$ pointwise, while $(C(s_2''),f_2'',C(t_2''))\notin \sim_{hp}$, because $s_2''$ and $C(s_2'')$ exist, but $t_2''$ and $C(t_2'')$ do not exist.
    \end{itemize}
\end{enumerate}
\end{proof}

\begin{theorem}[Equivalence and congruence for strong pomset bisimilarity]
\begin{enumerate}
  \item $\sim_p$ is an equivalence relation;
  \item If $P\sim_p Q$ then
  \begin{enumerate}
    \item $\alpha.P\sim_p \alpha.Q$, $\alpha$ is a free action;
    \item $P+R\sim_p Q+R$;
    \item $P\parallel R\sim_p Q\parallel R$;
    \item $(w)P\sim_p (w)Q$;
    \item $x(y).P\sim_p x(y).Q$.
  \end{enumerate}
\end{enumerate}
\end{theorem}

\begin{proof}
\begin{enumerate}
  \item It is sufficient to prove that $\sim_p$ is reflexivity, symmetry, and transitivity, we omit it.
  \item If $P\sim_p Q$, then
  \begin{enumerate}
    \item it is sufficient to prove the relation $R=\{(\alpha.P, \alpha.Q)|\alpha \textrm{ is a free action}\}\cup \textbf{Id}$ is a strong pomset bisimulation. It can be proved similarly to the proof of congruence for strong pomset bisimulation in CTC \cite{CTC}, we omit it;
    \item it is sufficient to prove the relation $R=\{(P+R, Q+R)\}\cup \textbf{Id}$ is a strong pomset bisimulation. It can be proved similarly to the proof of congruence for strong pomset bisimulation in CTC \cite{CTC}, we omit it;
    \item it is sufficient to prove the relation $R=\{(P\parallel R, Q\parallel R)\}\cup \textbf{Id}$ is a strong pomset bisimulation. It can be proved similarly to the proof of congruence for strong pomset bisimulation in CTC \cite{CTC}, we omit it;
    \item it is sufficient to prove the relation $R=\{((w)P, (w).Q)\}\cup \textbf{Id}$ is a strong pomset bisimulation. It can be proved similarly to the proof of congruence for strong pomset bisimulation in CTC \cite{CTC}, we omit it;
    \item it is sufficient to prove the relation $R=\{(x(y).P, x(y).Q)\}\cup \textbf{Id}$ is a strong pomset bisimulation. It can be proved similarly to the proof of congruence for strong pomset bisimulation in CTC \cite{CTC}, we omit it.
  \end{enumerate}
\end{enumerate}
\end{proof}

\begin{theorem}[Equivalence and congruence for strong step bisimilarity]
\begin{enumerate}
  \item $\sim_s$ is an equivalence relation;
  \item If $P\sim_s Q$ then
  \begin{enumerate}
    \item $\alpha.P\sim_s \alpha.Q$, $\alpha$ is a free action;
    \item $P+R\sim_s Q+R$;
    \item $P\parallel R\sim_s Q\parallel R$;
    \item $(w)P\sim_s (w)Q$;
    \item $x(y).P\sim_s x(y).Q$.
  \end{enumerate}
\end{enumerate}
\end{theorem}

\begin{proof}
\begin{enumerate}
  \item It is sufficient to prove that $\sim_s$ is reflexivity, symmetry, and transitivity, we omit it.
  \item If $P\sim_s Q$, then
  \begin{enumerate}
    \item it is sufficient to prove the relation $R=\{(\alpha.P, \alpha.Q)|\alpha \textrm{ is a free action}\}\cup \textbf{Id}$ is a strong step bisimulation. It can be proved similarly to the proof of congruence for strong step bisimulation in CTC \cite{CTC}, we omit it;
    \item it is sufficient to prove the relation $R=\{(P+R, Q+R)\}\cup \textbf{Id}$ is a strong step bisimulation. It can be proved similarly to the proof of congruence for strong step bisimulation in CTC \cite{CTC}, we omit it;
    \item it is sufficient to prove the relation $R=\{(P\parallel R, Q\parallel R)\}\cup \textbf{Id}$ is a strong step bisimulation. It can be proved similarly to the proof of congruence for strong step bisimulation in CTC \cite{CTC}, we omit it;
    \item it is sufficient to prove the relation $R=\{((w)P, (w).Q)\}\cup \textbf{Id}$ is a strong step bisimulation. It can be proved similarly to the proof of congruence for strong step bisimulation in CTC \cite{CTC}, we omit it;
    \item it is sufficient to prove the relation $R=\{(x(y).P, x(y).Q)\}\cup \textbf{Id}$ is a strong step bisimulation. It can be proved similarly to the proof of congruence for strong step bisimulation in CTC \cite{CTC}, we omit it.
  \end{enumerate}
\end{enumerate}
\end{proof}

\begin{theorem}[Equivalence and congruence for strong hp-bisimilarity]
\begin{enumerate}
  \item $\sim_{hp}$ is an equivalence relation;
  \item If $P\sim_{hp} Q$ then
  \begin{enumerate}
    \item $\alpha.P\sim_{hp} \alpha.Q$, $\alpha$ is a free action;
    \item $P+R\sim_{hp} Q+R$;
    \item $P\parallel R\sim_{hp} Q\parallel R$;
    \item $(w)P\sim_{hp} (w)Q$;
    \item $x(y).P\sim_{hp} x(y).Q$.
  \end{enumerate}
\end{enumerate}
\end{theorem}

\begin{proof}
\begin{enumerate}
  \item It is sufficient to prove that $\sim_{hp}$ is reflexivity, symmetry, and transitivity, we omit it.
  \item If $P\sim_{hp} Q$, then
  \begin{enumerate}
    \item it is sufficient to prove the relation $R=\{(\alpha.P, \alpha.Q)|\alpha \textrm{ is a free action}\}\cup \textbf{Id}$ is a strong hp-bisimulation. It can be proved similarly to the proof of congruence for strong hp-bisimulation in CTC \cite{CTC}, we omit it;
    \item it is sufficient to prove the relation $R=\{(P+R, Q+R)\}\cup \textbf{Id}$ is a strong hp-bisimulation. It can be proved similarly to the proof of congruence for strong hp-bisimulation in CTC \cite{CTC}, we omit it;
    \item it is sufficient to prove the relation $R=\{(P\parallel R, Q\parallel R)\}\cup \textbf{Id}$ is a strong hp-bisimulation. It can be proved similarly to the proof of congruence for strong hp-bisimulation in CTC \cite{CTC}, we omit it;
    \item it is sufficient to prove the relation $R=\{((w)P, (w).Q)\}\cup \textbf{Id}$ is a strong hp-bisimulation. It can be proved similarly to the proof of congruence for strong hp-bisimulation in CTC \cite{CTC}, we omit it;
    \item it is sufficient to prove the relation $R=\{(x(y).P, x(y).Q)\}\cup \textbf{Id}$ is a strong hp-bisimulation. It can be proved similarly to the proof of congruence for strong hp-bisimulation in CTC \cite{CTC}, we omit it.
  \end{enumerate}
\end{enumerate}
\end{proof}

\begin{theorem}[Equivalence and congruence for strongly hhp-bisimilarity]
\begin{enumerate}
  \item $\sim_{hhp}$ is an equivalence relation;
  \item If $P\sim_{hhp} Q$ then
  \begin{enumerate}
    \item $\alpha.P\sim_{hhp} \alpha.Q$, $\alpha$ is a free action;
    \item $P+R\sim_{hhp} Q+R$;
    \item $P\parallel R\sim_{hhp} Q\parallel R$;
    \item $(w)P\sim_{hhp} (w)Q$;
    \item $x(y).P\sim_{hhp} x(y).Q$.
  \end{enumerate}
\end{enumerate}
\end{theorem}

\begin{proof}
\begin{enumerate}
  \item It is sufficient to prove that $\sim_{hhp}$ is reflexivity, symmetry, and transitivity, we omit it.
  \item If $P\sim_p Q$, then
  \begin{enumerate}
    \item it is sufficient to prove the relation $R=\{(\alpha.P, \alpha.Q)|\alpha \textrm{ is a free action}\}\cup \textbf{Id}$ is a strongly hhp-bisimulation. It can be proved similarly to the proof of congruence for strongly hhp-bisimulation in CTC \cite{CTC}, we omit it;
    \item it is sufficient to prove the relation $R=\{(P+R, Q+R)\}\cup \textbf{Id}$ is a strongly hhp-bisimulation. It can be proved similarly to the proof of congruence for strongly hhp-bisimulation in CTC \cite{CTC}, we omit it;
    \item it is sufficient to prove the relation $R=\{(P\parallel R, Q\parallel R)\}\cup \textbf{Id}$ is a strongly hhp-bisimulation. It can be proved similarly to the proof of congruence for strongly hhp-bisimulation in CTC \cite{CTC}, we omit it;
    \item it is sufficient to prove the relation $R=\{((w)P, (w).Q)\}\cup \textbf{Id}$ is a strongly hhp-bisimulation. It can be proved similarly to the proof of congruence for strongly hhp-bisimulation in CTC \cite{CTC}, we omit it;
    \item it is sufficient to prove the relation $R=\{(x(y).P, x(y).Q)\}\cup \textbf{Id}$ is a strongly hhp-bisimulation. It can be proved similarly to the proof of congruence for strongly hhp-bisimulation in CTC \cite{CTC}, we omit it.
  \end{enumerate}
\end{enumerate}
\end{proof}

\subsection{Recursion}

\begin{definition}
Let $X$ have arity $n$, and let $\widetilde{x}=x_1,\cdots,x_n$ be distinct names, and $fn(P)\subseteq\{x_1,\cdots,x_n\}$. The replacement of $X(\widetilde{x})$ by $P$ in $E$, written $E\{X(\widetilde{x}):=P\}$, means the result of replacing each subterm $X(\widetilde{y})$ in $E$ by $P\{\widetilde{y}/\widetilde{x}\}$.
\end{definition}

\begin{definition}
Let $E$ and $F$ be two process expressions containing only $X_1,\cdots,X_m$ with associated name sequences $\widetilde{x}_1,\cdots,\widetilde{x}_m$. Then,
\begin{enumerate}
  \item $E\sim_p F$ means $E(\widetilde{P})\sim_p F(\widetilde{P})$;
  \item $E\sim_s F$ means $E(\widetilde{P})\sim_s F(\widetilde{P})$;
  \item $E\sim_{hp} F$ means $E(\widetilde{P})\sim_{hp} F(\widetilde{P})$;
  \item $E\sim_{hhp} F$ means $E(\widetilde{P})\sim_{hhp} F(\widetilde{P})$;
\end{enumerate}

for all $\widetilde{P}$ such that $fn(P_i)\subseteq \widetilde{x}_i$ for each $i$.
\end{definition}

\begin{definition}
A term or identifier is weakly guarded in $P$ if it lies within some subterm $\alpha.Q$ or $(\alpha_1\parallel\cdots\parallel \alpha_n).Q$ of $P$.
\end{definition}

\begin{theorem}
Assume that $\widetilde{E}$ and $\widetilde{F}$ are expressions containing only $X_i$ with $\widetilde{x}_i$, and $\widetilde{A}$ and $\widetilde{B}$ are identifiers with $A_i$, $B_i$. Then, for all $i$,
\begin{enumerate}
  \item $E_i\sim_s F_i$, $A_i(\widetilde{x}_i)\overset{\text{def}}{=}E_i(\widetilde{A})$, $B_i(\widetilde{x}_i)\overset{\text{def}}{=}F_i(\widetilde{B})$, then $A_i(\widetilde{x}_i)\sim_s B_i(\widetilde{x}_i)$;
  \item $E_i\sim_p F_i$, $A_i(\widetilde{x}_i)\overset{\text{def}}{=}E_i(\widetilde{A})$, $B_i(\widetilde{x}_i)\overset{\text{def}}{=}F_i(\widetilde{B})$, then $A_i(\widetilde{x}_i)\sim_p B_i(\widetilde{x}_i)$;
  \item $E_i\sim_{hp} F_i$, $A_i(\widetilde{x}_i)\overset{\text{def}}{=}E_i(\widetilde{A})$, $B_i(\widetilde{x}_i)\overset{\text{def}}{=}F_i(\widetilde{B})$, then $A_i(\widetilde{x}_i)\sim_{hp} B_i(\widetilde{x}_i)$;
  \item $E_i\sim_{hhp} F_i$, $A_i(\widetilde{x}_i)\overset{\text{def}}{=}E_i(\widetilde{A})$, $B_i(\widetilde{x}_i)\overset{\text{def}}{=}F_i(\widetilde{B})$, then $A_i(\widetilde{x}_i)\sim_{hhp} B_i(\widetilde{x}_i)$.
\end{enumerate}
\end{theorem}

\begin{proof}
\begin{enumerate}
  \item $E_i\sim_s F_i$, $A_i(\widetilde{x}_i)\overset{\text{def}}{=}E_i(\widetilde{A})$, $B_i(\widetilde{x}_i)\overset{\text{def}}{=}F_i(\widetilde{B})$, then $A_i(\widetilde{x}_i)\sim_s B_i(\widetilde{x}_i)$.

      We will consider the case $I=\{1\}$ with loss of generality, and show the following relation $R$ is a strong step bisimulation.

      $$R=\{(G(A),G(B)):G\textrm{ has only identifier }X\}.$$

      By choosing $G\equiv X(\widetilde{y})$, it follows that $A(\widetilde{y})\sim_s B(\widetilde{y})$. It is sufficient to prove the following:
      \begin{enumerate}
        \item If $G(A)\xrightarrow{\{\alpha_1,\cdots,\alpha_n\}}P'$, where $\alpha_i(1\leq i\leq n)$ is a free action or bound output action with $bn(\alpha_1)\cap\cdots\cap bn(\alpha_n)\cap n(G(A),G(B))=\emptyset$, then $G(B)\xrightarrow{\{\alpha_1,\cdots,\alpha_n\}}Q''$ such that $P'\sim_s Q''$;
        \item If $G(A)\xrightarrow{x(y)}P'$ with $x\notin n(G(A),G(B))$, then $G(B)\xrightarrow{x(y)}Q''$, such that for all $u$, $P'\{u/y\}\sim_s Q''\{u/y\}$.
      \end{enumerate}

      To prove the above properties, it is sufficient to induct on the depth of inference and quite routine, we omit it.
  \item $E_i\sim_p F_i$, $A_i(\widetilde{x}_i)\overset{\text{def}}{=}E_i(\widetilde{A})$, $B_i(\widetilde{x}_i)\overset{\text{def}}{=}F_i(\widetilde{B})$, then $A_i(\widetilde{x}_i)\sim_p B_i(\widetilde{x}_i)$. It can be proven similarly to the above case.
  \item $E_i\sim_{hp} F_i$, $A_i(\widetilde{x}_i)\overset{\text{def}}{=}E_i(\widetilde{A})$, $B_i(\widetilde{x}_i)\overset{\text{def}}{=}F_i(\widetilde{B})$, then $A_i(\widetilde{x}_i)\sim_{hp} B_i(\widetilde{x}_i)$. It can be proven similarly to the above case.
  \item $E_i\sim_{hhp} F_i$, $A_i(\widetilde{x}_i)\overset{\text{def}}{=}E_i(\widetilde{A})$, $B_i(\widetilde{x}_i)\overset{\text{def}}{=}F_i(\widetilde{B})$, then $A_i(\widetilde{x}_i)\sim_{hhp} B_i(\widetilde{x}_i)$. It can be proven similarly to the above case.
\end{enumerate}
\end{proof}

\begin{theorem}[Unique solution of equations]
Assume $\widetilde{E}$ are expressions containing only $X_i$ with $\widetilde{x}_i$, and each $X_i$ is weakly guarded in each $E_j$. Assume that $\widetilde{P}$ and $\widetilde{Q}$ are processes such that $fn(P_i)\subseteq \widetilde{x}_i$ and $fn(Q_i)\subseteq \widetilde{x}_i$. Then, for all $i$,
\begin{enumerate}
  \item if $P_i\sim_p E_i(\widetilde{P})$, $Q_i\sim_p E_i(\widetilde{Q})$, then $P_i\sim_p Q_i$;
  \item if $P_i\sim_s E_i(\widetilde{P})$, $Q_i\sim_s E_i(\widetilde{Q})$, then $P_i\sim_s Q_i$;
  \item if $P_i\sim_{hp} E_i(\widetilde{P})$, $Q_i\sim_{hp} E_i(\widetilde{Q})$, then $P_i\sim_{hp} Q_i$;
  \item if $P_i\sim_{hhp} E_i(\widetilde{P})$, $Q_i\sim_{hhp} E_i(\widetilde{Q})$, then $P_i\sim_{hhp} Q_i$.
\end{enumerate}
\end{theorem}

\begin{proof}
\begin{enumerate}
  \item It is similar to the proof of unique solution of equations for strong pomset bisimulation in CTC \cite{CTC}, we omit it;
  \item It is similar to the proof of unique solution of equations for strong step bisimulation in CTC \cite{CTC}, we omit it;
  \item It is similar to the proof of unique solution of equations for strong hp-bisimulation in CTC \cite{CTC}, we omit it;
  \item It is similar to the proof of unique solution of equations for strong hhp-bisimulation in CTC \cite{CTC}, we omit it.
\end{enumerate}
\end{proof} 

\section{Algebraic Theory}\label{at}

In this section, we will try to axiomatize $\pi_{tc}$, the theory is \textbf{STC} (for strongly truly concurrency).

\begin{definition}[STC]
The theory \textbf{STC} is consisted of the following axioms and inference rules:

\begin{enumerate}
  \item Alpha-conversion $\textbf{A}$.
  \[\textrm{if } P\equiv Q, \textrm{ then } P=Q\]
  \item Congruence $\textbf{C}$. If $p=Q$, then,
  \[\tau.P=\tau.Q\quad \overline{x}y.P=\overline{x}y.Q\]
  \[P+R=Q+R\quad P\parallel R=Q\parallel R\]
  \[(x)P=(x)Q\quad x(y).P=x(y).Q\]
  \item Summation $\textbf{S}$.
  \[\textbf{S0}\quad P+\textbf{nil}=P\]
  \[\textbf{S1}\quad P+P=P\]
  \[\textbf{S2}\quad P+Q=Q+P\]
  \[\textbf{S3}\quad P+(Q+R)=(P+Q)+R\]
  \item Restriction $\textbf{R}$.
  \[\textbf{R0}\quad (x)P=P\quad \textrm{ if }x\notin fn(P)\]
  \[\textbf{R1}\quad (x)(y)P=(y)(x)P\]
  \[\textbf{R2}\quad (x)(P+Q)=(x)P+(x)Q\]
  \[\textbf{R3}\quad (x)\alpha.P=\alpha.(x)P\quad \textrm{ if }x\notin n(\alpha)\]
  \[\textbf{R4}\quad (x)\alpha.P=\textbf{nil}\quad \textrm{ if }x\textrm{is the subject of }\alpha\]
  \item Expansion $\textbf{E}$.
  Let $P\equiv\sum_i \alpha_i.P_i$ and $Q\equiv\sum_j\beta_j.Q_j$, where $bn(\alpha_i)\cap fn(Q)=\emptyset$ for all $i$, and $bn(\beta_j)\cap fn(P)=\emptyset$ for all $j$. Then
  $P\parallel Q= \sum_i\sum_j (\alpha_i\parallel \beta_j).(P_i\parallel Q_j)+\sum_{\alpha_i \textrm{ comp }\beta_j}\tau.R_{ij}$.

  Where $\alpha_i$ comp $\beta_j$ and $R_{ij}$ are defined as follows:
\begin{enumerate}
  \item $\alpha_i$ is $\overline{x}u$ and $\beta_j$ is $x(v)$, then $R_{ij}=P_i\parallel Q_j\{u/v\}$;
  \item $\alpha_i$ is $\overline{x}(u)$ and $\beta_j$ is $x(v)$, then $R_{ij}=(w)(P_i\{w/u\}\parallel Q_j\{w/v\})$, if $w\notin fn((u)P_i)\cup fn((v)Q_j)$;
  \item $\alpha_i$ is $x(v)$ and $\beta_j$ is $\overline{x}u$, then $R_{ij}=P_i\{u/v\}\parallel Q_j$;
  \item $\alpha_i$ is $x(v)$ and $\beta_j$ is $\overline{x}(u)$, then $R_{ij}=(w)(P_i\{w/v\}\parallel Q_j\{w/u\})$, if $w\notin fn((v)P_i)\cup fn((u)Q_j)$.
\end{enumerate}
  \item Identifier $\textbf{I}$.
  \[\textrm{If }A(\widetilde{x})\overset{\text{def}}{=}P,\textrm{ then }A(\widetilde{y})= P\{\widetilde{y}/\widetilde{x}\}.\]
\end{enumerate}
\end{definition}

\begin{theorem}[Soundness]
If $\textbf{STC}\vdash P=Q$ then
\begin{enumerate}
  \item $P\sim_p Q$;
  \item $P\sim_s Q$;
  \item $P\sim_{hp} Q$.
\end{enumerate}
\end{theorem}

\begin{proof}
The soundness of these laws modulo strongly truly concurrent bisimilarities is already proven in Section \ref{stcb}.
\end{proof}

\begin{definition}
The agent identifier $A$ is weakly guardedly defined if every agent identifier is weakly guarded in the right-hand side of the definition of $A$.
\end{definition}

\begin{definition}[Head normal form]
A Process $P$ is in head normal form if it is a sum of the prefixes:

$$P\equiv \sum_i(\alpha_{i1}\parallel\cdots\parallel\alpha_{in}).P_i.$$
\end{definition}

\begin{proposition}
If every agent identifier is weakly guardedly defined, then for any process $P$, there is a head normal form $H$ such that

$$\textbf{STC}\vdash P=H.$$
\end{proposition}

\begin{proof}
It is sufficient to induct on the structure of $P$ and quite obvious.
\end{proof}

\begin{theorem}[Completeness]
For all processes $P$ and $Q$,
\begin{enumerate}
  \item if $P\sim_p Q$, then $\textbf{STC}\vdash P=Q$;
  \item if $P\sim_s Q$, then $\textbf{STC}\vdash P=Q$;
  \item if $P\sim_{hp} Q$, then $\textbf{STC}\vdash P=Q$.
\end{enumerate}
\end{theorem}

\begin{proof}
\begin{enumerate}
  \item if $P\sim_s Q$, then $\textbf{STC}\vdash P=Q$.
Since $P$ and $Q$ all have head normal forms, let $P\equiv\sum_{i=1}^k\alpha_i.P_i$ and $Q\equiv\sum_{i=1}^k\beta_i.Q_i$. Then the depth of $P$, denoted as $d(P)=0$, if $k=0$; $d(P)=1+max\{d(P_i)\}$ for $1\leq i\leq k$. The depth $d(Q)$ can be defined similarly.

It is sufficient to induct on $d=d(P)+d(Q)$. When $d=0$, $P\equiv\textbf{nil}$ and $Q\equiv\textbf{nil}$, $P=Q$, as desired.

Suppose $d>0$.

\begin{itemize}
  \item If $(\alpha_1\parallel\cdots\parallel\alpha_n).M$ with $\alpha_i(1\leq i\leq n)$ free actions is a summand of $P$, then $P\xrightarrow{\{\alpha_1,\cdots,\alpha_n\}}M$. Since $Q$ is in head normal form and has a summand $(\alpha_1\parallel\cdots\parallel\alpha_n).N$ such that $M\sim_s N$, by the induction hypothesis $\textbf{STC}\vdash M=N$, $\textbf{STC}\vdash (\alpha_1\parallel\cdots\parallel\alpha_n).M= (\alpha_1\parallel\cdots\parallel\alpha_n).N$;
  \item If $x(y).M$ is a summand of $P$, then for $z\notin n(P, Q)$, $P\xrightarrow{x(z)}M'\equiv M\{z/y\}$. Since $Q$ is in head normal form and has a summand $x(w).N$ such that for all $v$, $M'\{v/z\}\sim_s N'\{v/z\}$ where $N'\equiv N\{z/w\}$, by the induction hypothesis $\textbf{STC}\vdash M'\{v/z\}=N'\{v/z\}$, by the axioms $\textbf{C}$ and $\textbf{A}$, $\textbf{STC}\vdash x(y).M=x(w).N$;
  \item If $\overline{x}(y).M$ is a summand of $P$, then for $z\notin n(P,Q)$, $P\xrightarrow{\overline{x}(z)}M'\equiv M\{z/y\}$. Since $Q$ is in head normal form and has a summand $\overline{x}(w).N$ such that $M'\sim_s N'$ where $N'\equiv N\{z/w\}$, by the induction hypothesis $\textbf{STC}\vdash M'=N'$, by the axioms $\textbf{A}$ and $\textbf{C}$, $\textbf{STC}\vdash \overline{x}(y).M= \overline{x}(w).N$;
\end{itemize}

  \item if $P\sim_p Q$, then $\textbf{STC}\vdash P=Q$. It can be proven similarly to the above case.
  \item if $P\sim_{hp} Q$, then $\textbf{STC}\vdash P=Q$. It can be proven similarly to the above case.
\end{enumerate}
\end{proof} 

\section{Conclusions}\label{con}

This work is a mixture of mobile processes and true concurrency called $\pi_{tc}$, based on our previous work on truly concurrent process algebra -- a calculus for true concurrency CTC \cite{CTC} and an axiomatization for true concurrency APTC \cite{ATC}. 

$\pi_{tc}$ makes truly concurrent process algebra have the ability to model and verify mobile systems in a flavor of true concurrency, and can be used as a formal tool.

\label{lastpage}

\end{document}